\documentclass[11pt]{article}
\usepackage{graphicx, amsmath, amsthm, amssymb} 
\usepackage{xcolor}
\usepackage{algpseudocode}
\usepackage[margin=1in]{geometry}
\usepackage{setspace}
\usepackage{multirow}

\usepackage{authblk}
\usepackage{complexity}

\usepackage{tikz}
\usetikzlibrary{shapes.arrows,chains,decorations.pathreplacing}
\usetikzlibrary{arrows, decorations.markings, shapes}
\usetikzlibrary{positioning}

\usepackage{hyperref}
\usepackage[capitalize, nameinlink, noabbrev]{cleveref}
\newcounter{mc}
\newtheorem{theorem}{Theorem}
\newtheorem*{theorem*}{Theorem}
\newtheorem{claim}{Claim}
\newtheorem{lemma}[mc]{Lemma}

\usepackage{todonotes}
\usepackage{xcolor}
\usepackage{soul}
\definecolor{jj}{HTML}{faacec}
\definecolor{vp}{HTML}{66b2b2}

\usepackage{tcolorbox}
\tcbuselibrary{breakable}
\usepackage[margin=1in]{geometry}

\newcommand{\OPT}{\mathsf{OPT}}
\renewcommand{\C}{\mathsf{C}}

\renewcommand{\M}{\mathcal{M}}
\renewcommand{\G}{\mathcal{G}}

\newcommand{\T}{\mathcal{T}}
\renewcommand{\H}{\mathcal{H}}
\renewcommand{\A}{\mathcal{A}}
\renewcommand{\S}{\mathcal{S}}
\newcommand{\F}{\mathcal{F}}
\newcommand{\Fbig}{\F_{\text{big}}}
\newcommand{\Fsmall}{\F_{\text{small}}}

\newcommand{\N}{\mathsf{N}}

\newcommand{\copt}{\C^{\infty}}
\newcommand{\rcopt}{R(\copt/\varphi)}
\newcommand{\Aearly}{\A_{\text{early}}}
\newcommand{\Alate}{\A_{\text{late}}}
\newcommand{\tC}{\widetilde{\C}}

\title{Tight Bounds for Online Scheduling\\ in the One-Fast-Many-Slow Machines Setting}
\author{John Jeang}
\author{Vladimir Podolskii}
\affil{Tufts University}

\begin{document}

\maketitle
\begin{abstract}
    In the One-Fast-Many-Slow decision problem, introduced by Sheffield and Westover~\cite{itcs_NSAW_2025}, a scheduler, with access to one fast machine and infinitely many slow machines, receives a series of tasks and must decide how to allocate the work among its machines. The goal here is to achieve the minimal overhead of an online algorithm over the optimal offline algorithm that knows all tasks and their arrival times in advance. Three versions of this setting were considered: Instantly-committing schedulers that must assign a task to a machine immediately and irrevocably, Eventually-committing schedulers whose assignments are irrevocable but can occur anytime after a task arrives, and Never-committing schedulers that can interrupt and restart a task on a different machine. Sheffield and Westover constructed online algorithms in all three settings. In the Instantly-committing model the optimal competitive ratio is equal to 2. In the Eventually-committing model, they give an algorithm achieving a 1.678 competitive ratio, while a lower bound of 1.618 was established by Kuszmaul and Westover in~\cite{SPAA_SPDP}. In the Never-committing model, Sheffield and Westover show that the optimal competitive ratio lies in the interval [1.366, 1.5]. In the latter two settings, the exact values of the optimal competitive ratios were left as an open problem, moreover Kuszmaul and Westover~\cite{SPAA_SPDP} conjectured that the lower bound in the Eventually-committing model is tight.
    
    In this paper we resolve this problem by providing tight bounds for the competitive ratios in the Eventually-committing and Never-committing settings. For the Eventually-committing model we prove Kuszmaul and Westover's conjecture by providing an algorithm that achieves a competitive ratio equal to the previous lower bound of $\frac{1+\sqrt{5}}{2}\approx 1.618$. Notably, our algorithm occasionally leaves the fast machine free even in the presence of unassigned tasks, as was shown to be necessary for any improvement by Sheffield and Westover. For the Never-committing model we provide an explicit Task Arrival Process (TAP) lower bounding the competitive ratio to the previous upper bound of 1.5. Unlike the TAP used to prove the previous lower bound, which uses only two tasks, we construct a series of TAPs with a growing number of tasks that approach the upper bound of 1.5.
    
\end{abstract}

\newpage

\section{Introduction}

Online algorithms deal with streaming data in which decisions are made without complete knowledge of the entire problem instance. An online algorithm's performance is measured relative to that of the optimal \textit{offline} algorithm that knows the entire problem instance prior to all decisions. Optimizing this metric, known as the \textit{competitive ratio}, for online algorithms has a wide range of applications to scheduling, networking, and learning theory \cite{vaze2023online}. 

Scheduling algorithms aim to efficiently distribute incoming tasks to optimize for some goal (e.g. makespan, tardiness, fairness). Efficient scheduling algorithms have been well-studied in both the offline and online regime for their potential application to a variety of settings such as cloud computing, service, and industrial environments~\cite{karger1999scheduling,pinedo2022scheduling,cloud_scheduling}. Sometimes in online scheduling, tasks can take advantage of specialized hardware to expedite processing. However, when resources are limited, it is unclear how to prioritize tasks. One of the approaches to study this problem was suggested by Sheffield and Westover~\cite{itcs_NSAW_2025} when they introduced the One-Fast-Many-Slow Decision Problem (OFMS). In OFMS, an algorithm, with access to a single instance of accelerated hardware and infinitely many weak processors, receives tasks, whose sizes and arrival times are unknown in advance, in an online stream and must allocate resources to minimize the completion of the set of tasks as a whole.

The OFMS problem is motivated by its relevance to the Serial to Parallel Decision Problem (SPDP) \cite{SPAA_SPDP}, which addresses multiprocessor scheduling with parallelizable tasks. The SPDP question concerns itself with the scheduling of perfectly parallelizable tasks on $p$ identical processors. Indeed, Sheffield and Westover reduce SPDP to OFMS in the limit where the number of processors greatly outnumbers the number of tasks. The reduction works by isolating a small number of processors to be used only for the serial implementation of tasks and treating the rest of the processors as a single high performance collection for tasks implemented in parallel; for more information we refer the reader to the paper \cite{itcs_NSAW_2025}.

In both settings the tasks are unknown ahead of time and arrive in an online stream and the authors analyze bounds on the competitive ratios of online algorithms. In the case of OFMS the following bounds are known:

\begin{itemize}
    \item \textbf{Instantly-committing schedulers:} Tasks must be assigned to a specific machine immediately upon arrival and cannot be moved. The matching lower~\cite{SPAA_SPDP} and upper~\cite{itcs_NSAW_2025} bounds of 2 on the competitive ratio are known.
    \item \textbf{Eventually-committing schedulers:} Tasks can be assigned any time after arrival, but cannot be moved. The best known lower bound is $\frac{1+\sqrt{5}}{2} \approx 1.618$~\cite{SPAA_SPDP} and the best known upper bound is $\approx 1.678$~\cite{itcs_NSAW_2025}.
    \item \textbf{Never-committing schedulers:} Tasks can be moved and restarted on a different machine. The best known lower and upper bounds are $\frac{1+\sqrt{3}}{2} \approx 1.366$ and $1.5$ respectively~\cite{itcs_NSAW_2025}.
\end{itemize}

Thus, although in the Instantly-committing model the optimal competitive ratio is known, its value in the other two settings were left in~\cite{itcs_NSAW_2025} as an open problem. Moreover, it is conjectured in \cite{SPAA_SPDP} that the lower bound in the Eventually-committing model is tight.

\paragraph*{Results.}

In this paper we resolve the open question pertaining to the optimal competitive ratio of an online algorithm in the OFMS problem in the Eventually-committing and Never-committing settings. Specifically we prove the conjecture made in \cite{SPAA_SPDP} by providing an optimal Eventually-committing scheduler achieving a competitive ratio of $\frac{1 + \sqrt5}{2} = \varphi$ meeting the previous lower bound from \cite{SPAA_SPDP}. In the Never-committing model, we give an explicit task arrival process that lower bounds the competitive ratio of any online algorithm to 1.5, matching the previous upper bound from \cite{itcs_NSAW_2025}. 

\paragraph*{Techniques.}

Our main technical contribution is providing an optimal $\varphi$-competitive online algorithm, in the Eventually-committing setting. In \cite{itcs_NSAW_2025}, the authors provide the optimal \textit{non-procrastinating} algorithm for the OFMS problem. That is, their algorithm is optimal among algorithms that do not delay task assignments when the fast machine is free. Their analysis consist of bounding the amount of time spent by their algorithm's fast machine on tasks assigned \textit{differently} from the optimal algorithm. At a high level, their algorithm does excellent for many Task Arrival Processes (TAP for short), but struggles in the case where it misallocates a small task that should have been held onto in wait of more information from the TAP. 

Our algorithm modifies the previous algorithm by occasionally holding onto tasks even when the fast machine is free. The main technical work is showing that this modification does not incur too much delay as to hinder the strengths of the previous algorithm, but provides significant enough delay such that tasks are assigned more accurately. The key insight is making this delay independent of a tasks arrival time and solely dependent on the size of the task. This allows us to bound the time at which tasks will receive additional delay thereby allowing it to perform similarly to the previous algorithm in \cite{itcs_NSAW_2025}, late in the TAP, but significantly differently at the beginning of the TAP. From here we can use similar techniques to bound the amount of time spent, by our algorithm's fast machine, on misallocated tasks.

In the Never-committing setting, we come up with an explicit sequence of TAPs lower bounding the competitive ratio of online algorithms in this model to the previous upper bound of 1.5. These TAPs, consist of a growing number of tasks, differing significantly from the two-task TAP giving the previous $\frac{1 + \sqrt{3}}{2}$ lower bound.

\paragraph*{Other related works.}

Much work has been done in the field of scheduling in both the offline and online setting. In the offline setting, in which an algorithm knows the entire problem instance, optimal scheduling has been shown to be $\NP$-hard in a variety of settings (see, e.g.,~\cite{du1989complexity,blazewicz1986scheduling, karp}), and progress has been made towards approximation algorithms~\cite{mounie1999efficient, ludwig1994scheduling, turek1992approximate}. In \cite{heterogeneous_scheduling}, they examine heterogeneous hardware in the offline setting, and \cite{bachtler2020robust} deals with robust single-machine scheduling where jobs are known in advance but the release date of a job is given as an interval, rather than a point, in which the job arrives.

The online setting has typically been explored in the context of unknown jobs and arrival times, dealing with either a single processor \cite{guo2024online} or multiple identical processors \cite{albers1997better}. In \cite{hurink2007online}, they consider the problem of scheduling jobs that may require multiple machines to process. The case of stochastic scheduling concerns itself with tasks coming from a known \textit{distribution} of sizes and arrival times ~\cite{gupta2017stochastic}. Finally the case of flexible constraints has been considered in both fuzzy scheduling as well as in case studies such as the nurse scheduling problem~\cite{dubois2003fuzzy, gutjahr2007aco, legrain2020online}.

\paragraph{Outline.}

The rest of the paper is organized as follows. In \cref{sec:prelims}, we state the problem and provide the notation needed for the analysis. In \cref{sec:eventuallycommit}, we give a $\varphi$-competitive online algorithm for the Eventually-committing setting of OFMS. The analysis in this section is the bulk of our technical work, so we first begin by proving some general properties of our algorithm and the OFMS problem, and then proceed to evaluating our algorithm's performance. Finally in \cref{sec:nevercommit}, we give a procedure for generating a family of TAPs that result in a tight lower bound of 1.5 for online algorithms in the Never-committing setting. 

\section{Preliminaries}
\label{sec:prelims}

\subsection{The Problem}

In this section we formally state the, previously introduced, \textbf{One-Fast-Many-Slow Decision Problem} (OFMS), where an online algorithm with access to infinitely many slow machines and a single fast machines must assign tasks to machines with the goal of minimizing the total completion time (makespan) of the set of tasks as a whole.

We first define a task $\tau_i = (f_i, s_i, t_i)$, representing the runtime on the fast machine, runtime on a slow machine, and arrival time respectively. We enforce that $t_i \geq 0$ and $s_i \geq f_i > 0$. An instance of OFMS is a Task Arrival Process (TAP) $\T = \{\tau_1, \dots, \tau_n\}$, with $t_j \geq t_i$ for $j > i$. An online scheduling algorithm $\M$ incrementally receives portions of $\T$ upon the arrival of tasks. That is, $\M$ learns $\tau_i$ at $t_i$, whereas and offline scheduling algorithm $\N$ receives the entirety of $\T$ at the start.

A scheduling algorithm's goal is to assign tasks to be run on machines to minimize the completion time of $\T$. The completion time is defined to be the minimum time $t$ at which all tasks in $\T$ have finished running. 

In the OFMS problem, each machine can only run a single task at a time. A machine is said to be free in the \textit{closed} intervals in which it is not running any tasks (and thus there is a notion of ``the last time a machine is free"). We are interested in two previously defined models for this problem: 

\begin{enumerate}
    \item \textbf{Eventually-committing schedulers}: A scheduler may hold on to a task without immediately assigning/starting it, but once the task begins running, the assignment cannot be changed.

    \item \textbf{Never-committing schedulers}: A scheduler may hold on to a task without immediately assigning/starting it, and even after the task begins running, the task can be moved to start over on a different machine.
\end{enumerate}

In \cite{itcs_NSAW_2025}, Westover and Sheffield also define \textit{Instantly-committing} schedulers in which a scheduler must commit and queue a task to a particular machine upon arrival. In this case the competitive ratio of the optimal online algorithm is known to be 2.

We note that there are several algorithms, with possibly different assignments, that solve the offline version of OFMS. We describe one such algorithm, $\OPT$, that we'll use for the remainder of the paper.

Given an OFMS problem instance $\T$, $\OPT$ first calculates the completion time of the assignment, $q_0$, that allocates all tasks in $\T$ to the fast machine. Then for every task $\tau_i \in \T$, $\OPT$ calculates the completion time of an assignment $q_i$ defined as follows: for all $\tau_j \in \T$,

\begin{equation} \text{ $\tau_j$'s assignment} = 
    \begin{cases}
        \text{a slow machine if $s_j + t_j \leq s_i + t_i$} \\
        \text{the fast machine otherwise,}
    \end{cases}
\end{equation}

Let $K(q_i)$ denote the completion time of assignment $q_i$. $\OPT$ then simply uses the assignments associated with the minimum completion time $c = \min_i(K(q_i))$. Intuitively, $\OPT$ assigns a task to the slow machine whenever it can do so without increasing the completion time of $\T$. To see that this algorithm produces an optimal assignment, consider an optimal assignment $\G$. If $\G$ assigns all tasks in $\T$ to the fast machine then this is exactly the $q_0$ assignment, otherwise consider $\G$'s last-ending slow machine, working on $\tau_k$, that ends at $s_k + t_k$. First note that all of the slow machines in the $q_k$ assignment will finish before $s_k + t_k$. Second, notice that $\G$'s fast machine will contain every task that is on $q_k$'s fast machine because it contains every task $\tau_i$ such that $s_i + t_i > s_k +t_k$. Thus $q_k$'s fast machine finishes at least as quickly as $\G$'s fast machine and so $q_k$ must be an optimal assignment.

Importantly, in this algorithm we have that for all tasks $\tau_j$ such that $s_j + t_j \leq c$, $\OPT$ assigns $\tau_j$ to a slow machine, and we make frequent use of this property in our proofs.

\subsection{Notation}
 
At every given point in time $t$, there is an optimal offline allocation of just the subset of tasks in $\T$ that have arrival time less than or equal to $t$. We use $\OPT(t)$ to denote the optimal offline assignment of this subset of tasks and we use $\OPT(\infty)$ for the optimal offline assignment of the whole $\T$. We refer to the completion time of $\OPT(t)$ as $\C^t$, and to the total completion of $\OPT(\infty)$ as $\copt$. We sometimes drop the argument and use $\OPT$ to denote the final assignment when it is clear from context. Furthermore, we use $\tC^t$ to refer to the runtime of $\OPT(t)$'s fast machine. Finally, when talking about the optimal offline algorithm without reference to a specific time we just use $\OPT$. For an online algorithm $\M$, we use $\M_f$ to denote the runtime of $\M$'s fast machine after the entire TAP has arrived.

$R(P)$ will denote the earliest time at which an online algorithm $\M$, knows that the optimal offline algorithm must take at least $P$ time (e.g. by calculating $\C^t)$.  

For ease of communication, we define a notion of \textit{largeness}, where we say that a task $\tau_x$ is larger than $\tau_y$ if and only if $s_x + t_x > s_y + t_y$. Note that this notion does not care about the time it takes the tasks to run on the fast machine.

\section{Eventually-committing scheduler}
\label{sec:eventuallycommit}

\subsection{A $\varphi$-Competitive Eventually-committing Scheduler}

We define $\varphi = \frac{1 + \sqrt{5}}{2}$, the golden ratio. Throughout the paper we use the fact that $\varphi$ is the positive solution to the equation 
$\varphi = 1 + \frac{1}{\varphi}$.

In this section we provide a $\varphi$-competitive algorithm for the eventually-committing setting of OFMS. In \cite{itcs_NSAW_2025}, Sheffield and Westover define \textit{non-procrastinating} algorithms as those that only delay the machine assignment of tasks if the fast machine is occupied. They show that a competitive ratio of $\approx1.678$ is optimal for non-procrastinating algorithms, and thus to achieve a competitive ratio of $\varphi \approx 1.618$ our algorithm needs to occasionally delay assigning a task to a machine even if the fast machine is free. To do so we first define the notion of an \textit{eligible} task. A task $\tau_i$ is eligible at time $t$ if 
    \begin{equation}
     t \geq \frac{f_i}{\varphi}.   
    \end{equation}
Our algorithm will only assign eligible tasks to the fast machine and therefore may sometimes procrastinate. 

Finally we note that an online algorithm $\M$ can (efficiently) calculate $\OPT(t)$ and $\C^t$ at time $t$, by running the algorithm for $\OPT$ on the tasks that have arrived at or before time $t$. This gives $\M$ a lower bound for $\copt$ at time $t$. We say that $\M$ \textit{safely} assigns a task to a slow machine if the finishing time of that machine is guaranteed to be within the $\varphi$-competitive ratio with the completion time of the optimal offline algorithm.

We now give the formal description of our Eventually-committing scheduler $\H$.

\paragraph{Algorithm $\H$:} At $t_0=0$, $\H$ initializes an empty standby set $\S$. Whenever a task $\tau_i$ arrives, it is immediately added to $\S$. Additionally $\H$ continuously evaluates all tasks in $\S$ at all times $t$ as follows:
\begin{enumerate}

    \item If $s_i + t \leq \varphi \C^t$, then start $\tau_i$ on an unused slow machine and remove it from $\S$.
    \item If \textit{all} of the following conditions hold: 
    
    \begin{enumerate}
        \item the fast machine is free or just finished a task,
        \item $\tau_i$ has the largest $s_i + t_i$ value among all tasks in $\S$, 
        \item $t \geq \frac{f_i}{\varphi}$ ($\tau_i$ is eligible),
    \end{enumerate}
    then start $\tau_i$ on the fast machine and remove it from $\S$.
\end{enumerate}

We refer to $\S$ as the \textit{standby} set, and denote $\S_t$ to be the set of task on standby at time $t$.

At a high level, $\H$ first does a ``slow-check" to see if a task can safely be assigned to a slow machine. If this check fails, $\H$ then does a ``fast-check" to see if the task can be assigned to the fast machine. We will frequently use these check conditions to characterize the performance of $\H$.

Although $\H$ makes these checks continuously, we observe the checks need only to be made during specific events. Notice that the slow-checks only need to occur upon the arrival of a task, because $s_i + t$ only increases in $t$ while $\varphi\C^t$ does not change in between the arrival of tasks. The fast-checks only need to occur upon the arrival of a task, when a task on the fast machine finishes, or when the largest task in $\S$ becomes eligible. Therefore, the algorithm can actually be made discrete without changing its behavior because there are only finitely many events warranting checks. This observation allows us to pick $\delta$ small enough such that for any event happening at time $t$, there are no events between $t-\delta$ and $t$. From this we can use the notation 

\begin{equation}
    t^- = t - \delta
    \label{eq:tminus}
\end{equation}
for what can be thought of as ``the time just before $t$".

The eligibility condition, (c) in the fast-check, is what causes $\H$ to sometimes procrastinate. However notice that eligibility is concerned with the absolute time $t$ but not the arrival time $t_i$ of a particular task $\tau_i$. Intuitively, this is how $\H$ manages to procrastinate mostly at the beginning of the TAP, but not too much near the end of the TAP.

Further notice that the eligibility condition is not present in the slow-check. This is because there are infinitely many slow machines, so we are not worried about keeping any of them free and thus $\H$ simply runs tasks on slow machines whenever it can safely do so. 

Finally, we emphasize that the conditions in the fast check must be met simultaneously. To run a task $\tau_i$ on the fast machine at time $t$, it is not sufficient for $\tau_i$ to merely be the largest among all eligible tasks in $\S_t$; it must be the largest among all tasks in $\S_t$. If there is a largeness tie, and at least one task participating in the tie is eligible, then $\H$ assigns one of the eligible tasks to the fast machine, breaking ties arbitrarily among the eligible tasks participating in the tie. 

At this point it is also useful to introduce the following sets of tasks that will help us with our analysis. For a given TAP $\T$ we define the following:

\begin{itemize}
    \item $\A$ (``Actual"): The set of tasks that both $\H$ and $\OPT(\infty)$ assign to the fast machine;\,
    \item $\F$ (``Fake"): The set of tasks that $\H$ assigns to the fast machine but $\OPT(\infty)$ assigns to a slow machine;\,
    \item $\F_{\text{big}} \subset \F$: The set of tasks in $\F$ such that $s_i + t_i > \frac{\copt}{\varphi}$;
    \item $\F_{\text{small}} \subset \F$: The set of tasks in $\F$ such that $s_i + t_i \leq \frac{\copt}{\varphi}$.
\end{itemize}

To avoid confusion, it is also important to note some subtleties. Firstly, we note that $\A \cup \F$ does not capture any of the tasks that $\H$ assigns to its slow machine; most of the time we simply do not care about these tasks since they finish within a $\varphi$ competitive ratio of $\OPT(\infty)$'s completion time, by the design of the algorithm. However, there is no explicit mechanism that prohibits the possibility of a task $\tau$ from being assigned to $\OPT(\infty)$'s fast machine, while being assigned to the slow machine by $\H$. While these tasks are not relevant for most of our analysis, it is important to remember that they exist since they can show up in standby set $\S$ and affect $\OPT$'s fast machine.

Finally, we note that since $\H$ only safely assigns tasks to slow machines, then to prove that $\H$ is $\varphi$-competitive, it suffices to consider just the performance of its fast machine.

\subsection{Initial Observations and Properties}

Before analyzing the performance of our algorithm, it is useful to familiarize ourselves with some of the basic properties of $\H$, $\OPT$, $\Fbig$, $\Fsmall$ and $\A$ that will greatly simplify our analysis. We begin by proving some properties of $\OPT$. 

Note that for $t_1 < t_2$ it is possible for $\OPT(t_1)$ to assign a task $\tau_i$ to the fast machine while $\OPT(t_2)$ assigns $\tau_i$ to a slow machine (note that $\OPT$ is an offline algorithm, and so this swapped assignment is retroactive and thus effective at $t_i$ and not $t_2$). This is because $\OPT(t_2)$ may benefit from freeing up space on the fast machine for other tasks arriving at $t_1 < t \leq t_2$. In this case we say that $\OPT$ has \textit{swapped} $\tau_i$'s assignment. This next lemma tells us that tasks only swap from the fast machine to a slow machine.

\begin{lemma}
    If $\tau_i$ is assigned to a slow machine by $\OPT(t_1)$, then it will be assigned to a slow machine by $\OPT(t_2)$ for all $t_2 > t_1$.
    \label{lem:noslowswap}
\end{lemma}

\begin{proof}
    Consider a task $\tau_i$ assigned to a slow machine by $\OPT(t_1)$. Let $t_2$ be any time greater than $t_1$. We have

    \begin{equation}
        s_i + t_i \leq \C^{t_1} \leq \C^{t_2},
        \label{eq:eqslow}
    \end{equation}
    where the first inequality comes from the fact that $\OPT(t_1)$'s completion time must be at least the completion time of $\tau_i$. The second inequality comes from the fact that all task parameters are non-negative and thus $\C^t$ is monotone in $t$.

    Since $\OPT(t)$ assigns a task $\tau_i$ to a slow machine whenever $s_i + t_i \leq \C^t$, we conclude that $\tau_i$ must be assigned to a slow machine by $\OPT(t_2)$.
\end{proof}

\cref{lem:noslowswap} helps us further characterize some of $\OPT$'s behavior. 

\begin{lemma}
    Each task $\tau_i \in \A \cup\F$ is assigned to the fast machine by $\OPT(t_i)$. 
    \label{lem:AFstartfast}
\end{lemma}

\begin{proof}
    First consider a task $\tau_a \in \A$. We know that $\OPT(\infty)$ assigns $\tau_a$ to the fast machine, and from \cref{lem:noslowswap} we know that $\OPT$ does not swap task assignments from slow to fast. Therefore $\tau_a$ must be assigned to the fast machine the entire time since its arrival at $t_a$.

    Now consider a task $\tau_z \in \F$. Suppose for the sake of contradiction that $\OPT(t_z)$ assigns $\tau_z$ to the slow machine. Then during $\H$'s slow-check we have 

    \begin{equation}
        \frac{s_z + t_z}{\C^{t_z}} \leq \frac{s_z + t_z}{\underbrace{s_z + t_z}_{\text{$\OPT$'s completion of $\tau_z$}}} = 1 < \varphi.
    \end{equation}
    
    This means that $\H_f$ would also assign $\tau_z$ to a slow machine at $t_z$. However this is a contradiction since $\H$ eventually assigns all tasks in $\F$ to the fast machine, and $\H$ is an Eventually-committing scheduler that cannot change the assignment of tasks. Therefore $\OPT(t_z)$ assigns $\tau_z$ to the fast machine.
\end{proof}

Finally, in our next lemma, we build a classification of the tasks assigned to $\H$'s fast machine.

\begin{lemma}

    For any $\tau_a \in \A$, $\tau_b \in \Fbig$, and $\tau_m \in \Fsmall$ we have the following:

    \begin{enumerate}
        \item $s_m + t_m \leq \frac{\copt}{\varphi} < s_b + t_b \leq \copt <  s_a + t_a$,
        \item $f_b \leq \frac{\copt}{\varphi}$, 
        \item $f_m \leq \frac{\copt}{\varphi^2}$.
    \end{enumerate}
    \label{lem:setprops}
\end{lemma}

\begin{proof}
    The first two inequalities from item 1  come from the definition of $\Fsmall$ and $\Fbig$. The third inequality from item 1 comes from the fact that $\tau_b$ is ultimately put on one of $\OPT$'s slow machines and the completion times of all of $\OPT$'s slow machines are bounded by $\OPT$'s total completion time. Finally the last inequality from item 1 comes from the fact that $\tau_a$ ultimately goes on $\OPT(\infty)$'s fast machine, and therefore must be unable to go on $\OPT$'s slow machine without increasing the completion time since $\OPT$ puts tasks on a slow machine whenever possible.
    
    To see items 2 and 3, it suffices to show that for any $\tau_i \in \Fbig \cup \Fsmall$ we have that $f_i < \frac{s_i}{\varphi}$, and then we can just apply item 1. Consider $\tau_i \in \Fbig \cup \Fsmall$. To show $f_i < \frac{s_i}{\varphi}$, we recall that by  \cref{lem:AFstartfast} we know $\OPT(t_i)$ assigns $\tau_i$ to the fast machine, but $\H$ doesn't assign $\tau_i$ to a slow machine at this time so we have 
    
    \begin{equation}  
      \varphi < \frac{s_i + t_i}{\C^{t_i}} \leq \frac{s_i + t_i}{\tC^{t_i}} \leq \frac{\overbrace{s_i + t_i}^{\tau_i \text{ on $\H$'s slow machine}}}
    {\underbrace{f_i + t_i}_{\text{minimum time on $\OPT$'s fast machine}}}
    \implies f_i \leq \frac{s_i}{\varphi}.
    \label{eq:slowcheckfail}
    \end{equation}
\end{proof}

With these initial observations in hand we can proceed to the main analysis.

\subsection{Analysis}

Our goal is to prove the following theorem.

\begin{theorem}\label{thm:competitive}
    $\H$ is a $\varphi$-competitive Eventually-committing scheduler.
\end{theorem}

Recall that we only have to concern ourselves with the performance of $\H$'s fast machine. Our first strategy for analyzing $\H_f$ will be to characterize the tasks from $\F$. These can be thought of as misassigned ``extra work" that contribute to $\H_f$ since the optimal offline algorithm assigns them to slow machines once the entire TAP has arrived. Our second strategy will be to bound the amount of time that $\H$'s fast machine spends procrastinating. 

We will need several key lemmas to combine these two ideas and bound $\H_f$. We start with proving that misassigned tasks that $\H$ runs on the fast machine, cannot arrive too late. 

\begin{lemma}
    All tasks $\tau_i \in \F$ arrive (strictly) before $\rcopt$.
    \label{lem:nolatearrive}
\end{lemma}

\begin{proof}
    All tasks $\tau_i \in \F$ are run on the fast machine by $\H$, but ultimately on a slow machine by $\OPT$ once the entire TAP has arrived. This implies 

    \begin{equation}
        \varphi \C^{t_i} < s_i + t_i \leq \copt \leq \varphi \C^{\rcopt}.
    \end{equation}

    The first inequality comes from the fact that $\H$ doesn't immediately assign $\tau_i$ to a slow machine. The second inequality comes from the fact that the runtime of $\tau_i$ cannot exceed the total runtime of the entire TAP since $\OPT$ ultimately runs $\tau_i$ on a slow machine. The last inequality comes from the fact that $\C^{R(\copt/\varphi)} \geq \copt/\varphi$ by definition of $R$. Thus $t_i < \rcopt)$. 
\end{proof}

For further analysis we need to define a few key sets of tasks. Let $\tau_l$ be the last task to arrive strictly before $\rcopt$:

\begin{itemize}
    \item $\alpha = \S_{t_l} \cap \A$,
    \item $\beta = \S_{t_l} \cap \Fbig$,
\end{itemize}

Furthermore, for $t \geq \rcopt$ we define
\begin{equation}
 \beta^{t} =  \S_t \cap \Fbig.
\end{equation}

Importantly, we note that for $t_2 \geq t_1 \geq \rcopt$, 

\begin{equation}
    \beta^{t_2} \subseteq \beta^{t_1} \subseteq \beta,
\end{equation}
because as time progresses for $t>\rcopt$, no additional tasks from $\F$ arrive (by \cref{lem:nolatearrive}) but some may be removed from standby. 

For any set $X$ we define $\Sigma_X = \sum_{i \in X}f_i$, the total load $X$ can add to the fast machine if all tasks in $X$ are run successively without pause. Using this notation, we see that $\Sigma_{\beta^{t_2}} \leq \Sigma_{\beta^{t_1}} \leq \Sigma_{\beta}$. 

In our next lemma, we show that $\Sigma_\alpha$ and $\Sigma_\beta$ cannot be too large.

\begin{lemma}
    $\Sigma_{\alpha} + \Sigma_{\beta} < \frac{\copt}{\varphi}$.
    \label{lem:boundalphabeta}
\end{lemma}

\begin{proof}
    Observe that any task $\tau_i \in \alpha \cup \beta$ must be assigned to the fast machine by $\OPT(t_l)$  or else $\C^{t_l}$ would exceed $\copt/\varphi$, by Lemma 3, which is not possible since $t_l < \rcopt$. Thus we have 

    \begin{equation}
     \Sigma_{\alpha} + \Sigma_{\beta} < \tC^{t_l} \leq \C^{t_l} < \frac{\copt}{\varphi}.   
    \end{equation}
\end{proof}

From this we define

\begin{equation}
    0 < \varepsilon  = \frac{\copt}{\varphi} - (\Sigma_{\alpha} + \Sigma_{\beta}).
    \label{eq:boundstandby}
\end{equation}

Our next lemma bounds how long tasks from $\A$ can stay ineligible. 

\begin{lemma}
    Let $\tau_a$ be a task in $\A$. If either $\tau_a \in \alpha$ or both $t_a \geq \rcopt$ and $\beta^{t_a} \neq \emptyset$, then $\tau_a$ is eligible, for $\H$'s fast machine, for all $t \geq \max(t_a, \Sigma_\alpha, \varepsilon + \Sigma_\beta - \Sigma_{\beta^{t_a}})$.
    \label{lem:early_eligible}
\end{lemma}

\begin{proof}
    Consider a task $\tau_a \in \A$ arriving at $t_a \geq \rcopt$ with $\beta^{t_a} \neq \emptyset$. If the task is eligible upon arrival then we're done. Otherwise we have that $t_a < \frac{f_a}{\varphi} $.

    From \cref{lem:AFstartfast}, we know that $\tau_a$ is assigned to the fast machine by $\OPT(t_a)$. Furthermore we know that for any $\tau_i \in \beta^{t_a}$, $\H$ does not start $\tau_i$ on a slow machine at $t_a$ because all tasks in $\beta^{t_a}$ are ultimately run on $\H$'s fast machine. Thus similarly to \eqref{eq:slowcheckfail} we have

    \begin{equation}
        \varphi < \frac{\overbrace{s_i + t_a}^{\tau_i \text{  on $\H$'s slow machine at $t_a$}}}{\C^{t_a}} \leq \frac{s_i + t_a}{\tC^{t_a}} \leq \frac{s_i + t_a}{f_a + t_a}
        \implies \frac{s_i}{\varphi} > f_a.
        \label{eq:weak}
    \end{equation}

    This tells us that $t_a < \frac{f_a}{\varphi} < \frac{s_i}{\varphi^2} \leq \frac{s_i + t_i}{\varphi^2} $. From this we can see that if any $\tau_i \in \beta^{t_a}$ were assigned to a slow machine by $\OPT(t_a)$ then $\H$ would assign $\tau_i$ to a slow machine because then we would have the following slow-check for $\tau_i$:

    \begin{equation}
        \frac{s_i +t_a}{\C^{t_a}} \leq \frac{(s_i + t_i) + \frac{s_i + t_i}{\varphi^2}}{\underbrace{s_i + t_i}_{\tau_i \text{ on $\OPT(t_a)$'s slow machine} }} = 1 + \frac{1}{\varphi^2} \leq \varphi.
    \end{equation}

This gives us 

\begin{equation}
    \Sigma_\alpha + \Sigma_{\beta^{t_a}} + f_a \leq \tC^{t_a}    
\end{equation}
because all of $\alpha \cup \beta^{t_a} \cup \{\tau_a\}$ are on $\OPT(t_a)$'s fast machine. Furthermore, recall that for any $\tau_i \in \beta^{t_a}$, we have $s_i \leq \copt$ from \cref{lem:setprops} because $\beta^{t_a} \subseteq \Fbig$. Taken together we can get a stronger version of \eqref{eq:weak}. That is, we now have that for any $\tau_i \in \beta^{t_a}$

\begin{equation}
    \varphi < \frac{s_i + t_a}{\tC^{t_a}}   \leq \frac{\copt + \frac{f_a}{\varphi}}{\Sigma_\alpha + \Sigma_{\beta^{t_a}} + f_a} 
    =\frac{\copt + \frac{f_a}{\varphi}}{(\frac{\copt}{\varphi} - \varepsilon) + (\Sigma_{\beta^{t_a}} - \Sigma_\beta) + f_a}.
\end{equation}

Rearranging the inequality from the first and last expression gives us $\varphi(\varepsilon + \Sigma_\beta - \Sigma_{\beta^{t_a}})> (\varphi - \frac{1}{\varphi})f_a = f_a$ and thus $\tau_a$ will be eligible for $t \geq \varepsilon + \Sigma_\beta - \Sigma_{\beta^{t_a}}$.

Now consider a task $\tau_a \in \alpha$. It is clear that it will be eligible for all $t \geq \frac{\Sigma_\alpha}{\varphi}$ since $f_a \leq \Sigma_\alpha$. Therefore $\tau_a$ will certainly be eligible for $t \geq \Sigma_\alpha$. 

Combining these two results gives us the lemma.
\end{proof}

Next we have a useful lemma that simplifies our analysis. For this we remind the reader of the notion of $t^-$, defined in \eqref{eq:tminus}, as the timestep ``just before $t$".

\begin{lemma}
    Assume there exist some time before $\H_f$ when $\H$'s fast machine is free, and let $t$ be the last of such times. If all of the following conditions hold 
    
    \begin{enumerate}
        \item $t \geq \rcopt$,
        \item $\Fsmall \cap S_{t^-} = \emptyset$,
        \item $\A \cap S_{t^-} \neq \emptyset$,
        \item All tasks in $\A \cap S_{t^-}$ are eligible (at $t^-$),
        
    \end{enumerate}
    then $\H_f \leq \varphi\copt$.
    \label{lem:weirdgood}
\end{lemma}

\begin{proof}
    Suppose that the four conditions of the lemma hold.

    Observe that conditions 1 and 2 imply that the remaining contribution to $\H_f$ from $\F$-tasks after time $t$ is bounded by $\Sigma_\beta$ because we can assume by \cref{lem:nolatearrive} that no additional tasks from $\F$ arrive after $\rcopt$. 

    Conditions 3 and 4 imply that there is at least one large task on standby preventing the tasks from $\A$ from running on $\H$'s fast machine at $t^-$. This is because $\H$'s fast machine is free at $t^-$ if $t$ is the last time it is free. Thus there exists a non-empty subset of ineligible tasks $V \subset \S_{t^-}$ such that $V \cap \A = \emptyset$ and that for every task $\tau_v \in V$ and $\tau_a \in \A \cap S_{t^-}$ we have

    \begin{equation}
        t_v + s_v > t_a + s_a > \copt,
    \end{equation}
    where the last inequality comes from \cref{lem:setprops}. This implies that all tasks in $V$ are assigned to the fast machine by $\copt$, otherwise $\OPT$'s final completion time would be larger than $\copt$. Let $\tau_w \in V$ be such that 
    
    \begin{equation}
     f_w = \max_{\tau_i \in V}(f_i).   
     \label{eq:tau_define}
    \end{equation}
     
    In other words, $\tau_w$ is the task with the longest fast runtime in $V$. Since $\tau_w \notin \A$, it does not contribute to $\H_f$. We can now characterize the contribution to $\H_f$ from $\A$ as follows:

    \begin{equation}
    \copt \geq \tC^{\infty} \geq f_w +\sum_{\tau_i \in \A}f_i \implies \copt - f_w \geq \sum_{\tau_i \in \A}f_i.
    \end{equation}

    Since $\tau_w$ is ineligible, then $t < \frac{f_w}{\varphi}$. Since $t$ is the last time $\H$'s fast machine is free, we have 

    \begin{equation}
        \H_f \leq t + \overbrace{\Sigma_\beta}^{\text{work from $\F$}} + \underbrace{\sum_{\tau_i \in \A}f_i}_{\text{work from $\A$}} < \frac{f_w}{\varphi} + \frac{\copt}{\varphi} + (\copt - f_w) \leq \left(1 + \frac{1}{\varphi}\right)\copt = \varphi \copt,
    \end{equation}
    where in the second inequality we use the fact that $\Sigma_\beta < \frac{\copt}{\varphi}$ from \cref{lem:boundalphabeta}.
\end{proof}

Next we prove a lemma that narrows down our analysis.

\begin{lemma}
    If $\H$ starts a task $\tau \in \F_{\text{small}}$ at any time $t \geq R(\copt/\varphi)$ then $\H_f \leq \varphi\copt$.
    \label{lem:nolatesmall}
\end{lemma}

\begin{proof}

    First observe that if any task $\tau_m \in \Fsmall$ is waiting on standby at $t = \copt$, then $\H$ will be aware that $\OPT$ takes at least $\copt$ time and will run $\tau_m$ on a slow machine because we have

    \begin{equation}
        \underbrace{s_m + t}_{\tau_m \text{ on $\H$'s slow machine}} = s_m + \copt < s_m + t_m + \copt \leq \frac{\copt}{\varphi} + \copt \leq \varphi\copt,
    \end{equation}
    which is impossible since tasks from $\Fsmall$ are all run on $\H$'s fast machine. Since no tasks can arrive after $\copt$, this implies that no task from $\Fsmall$ will start after $\copt$.

    Now let $t^* \geq \rcopt$ be the last time that $\H$ starts a task $\tau_m \in \Fsmall$. Since $\H$ prioritizes tasks with larger $s_i + t_i$, then by \cref{lem:setprops} there must be no tasks $\tau \in \Fbig \cup \A$ in $\S_{t^*}$. By \cref{lem:nolatearrive}, no tasks from $\F$ will arrive from this point forward, and thus all remaining work contributing to $\H_f$ is from $\A$. Further we note that from this point forward at most $\copt - t^*$ amount of work can arrive from $\A$ because $\OPT$ must complete these tasks on the fast machine; notice that because $t^* \leq \copt$ this amount of work is valid. 

    Since for $\tau_a \in \A$ we have $f_a < \copt$, we know that all tasks from $\A$ will be eligible by $\frac{\copt}{\varphi}$. Thus if $\H$ runs continuously from $\max(t^* + f_m, \frac{\copt}{\varphi})$ then we have 

    \begin{align}
        \H_f &\leq \max(t^* + f_m, \frac{\copt}{\varphi}) + (\copt - t^*) \\
        &= \max(\copt + f_m, (1 + \frac{1}{\varphi})\copt - t^*) \\
        &\leq \max((1 + \frac{1}{\varphi^2})\copt, (1 + \frac{1}{\varphi})\copt) \\
        &= \varphi\copt,
    \end{align}
    where in the last inequality we used the fact that $f_m \leq \frac{\copt}{\varphi^2}$ from \cref{lem:setprops}.

    Otherwise $\H$'s fast machine is free at some point after $\max(t^* + f_m, \frac{\copt}{\varphi})$. Let $t > \max(t^* + f_m, \frac{\copt}{\varphi})$ be the last time $\H$'s fast machine is free.
    
    Suppose we have $\A \cap S_{t^-} = \emptyset$. Since $\F \cap \S_t = \emptyset$, this implies that from $t$ onward $\H$'s fast machine only works on tasks from $\A$, with arrival time greater than or equal to $t$. Since $\H$'s fast machine does this work continuously it must finish processing the set $\A$ at least as quickly as $\OPT$'s fast machine. This gives us $\H_f \leq \copt$.
    
    Otherwise we have that $\A \cap S_{t^-} \neq \emptyset$. We note that since all tasks in $\A$ are eligible at $\max(t^* + f_m, \frac{\copt}{\varphi}) < t$, then they are all also eligible at $t^-$. Thus we have all the necessary conditions to invoke \cref{lem:weirdgood}. 
\end{proof}

\cref{lem:nolatesmall} basically tells us that we can assume that all tasks in $\Fsmall$ start before $\rcopt$. Another way to think about it is that no tasks from $\Fsmall$ are in $\S_t$ for any $t \geq \rcopt$.

The next lemma strengthens \cref{lem:weirdgood} and tells us that $\H$ performs well, unconditionally, if its fast machine is ever free at a late time.

\begin{lemma}
    If $\H$'s fast machine is free at any time $\rcopt \leq t < \H_f$ then $\H_f \leq \varphi\copt$.
    \label{lem:lategaps}
\end{lemma}

\begin{proof}

    If $t \geq \copt$ then all tasks have arrived and all tasks are eligible. To see the latter observe that for any task $\tau_i$, we have
    \begin{equation}
     f_i \leq \copt < \varphi \copt \leq \varphi t,   
    \end{equation}
    where the first inequality comes from the fact that $\OPT$ finishes processing $\tau_i$ on some machine prior to $\copt$. In this case if $\H$ is free, then there must not be any tasks on standby and thus $\H$'s fast machine has already completed and thus $t \geq H_f$ and the lemma is vacuously true. So we only focus on the case where $t < \copt$.

    Consider the last time $t \geq \rcopt$ such that $\H$'s fast machine is free. 
        
    Due to \cref{lem:nolatesmall}, we can assume that $\Fsmall \cap \S_t = \emptyset$. Thus we have that 
    \begin{equation}
        \Sigma_{(\Fsmall \cup \Fbig) \cap S_t} = \Sigma_{\Fbig \cap S_t} = \Sigma_{\beta^t}.    
        \label{eq:boundst}
    \end{equation}
    
    In other words, the total contribution of $\F$ tasks on standby, to $\H_f$, is bounded by $\Sigma_{\beta^t}$.

    \paragraph{Case 1 ($\S_{t^-} \cap \A = \emptyset$):}

    In this case, from time $t$, the total contribution from $\A$-tasks to $\H_f$ is at most $\copt - t$ and the total contribution from $\F$-tasks is $\Sigma_{\beta^t} \leq \Sigma_{\beta}$. Therefore we have

    \begin{equation}
        \H_f \leq t + (\copt - t) + \Sigma_{\beta^t} \leq t + (\copt - t) + \Sigma_\beta \leq (1 + \frac{1}{\varphi})\copt = \varphi\copt,
    \end{equation}
    where the last inequality comes from \cref{lem:boundalphabeta}.

    \paragraph{Case 2 ($\S_{t^-} \cap \A \neq \emptyset$ and all tasks from $\S_{t^-} \cap \A$ are eligible):}

    This case gives us all the conditions necessary to invoke \cref{lem:weirdgood}.

    \paragraph{Case 3 ($\S_{t^-} \cap \A \neq \emptyset$ but some tasks in $\S_{t^-} \cap \A$ are ineligible):}

    Let us define $\tau_a$ to be the last arriving task in $\S_{t^-} \cap \A$ such that $\rcopt \leq t_a < t$. If $\tau_a$ doesn't exist then the contribution from $\S_t$ to $\H_f$ is bounded by $\Sigma_\alpha + \Sigma_\beta < \frac{\copt}{\varphi}$. Since the maximum load for all future tasks from $\A$ is $\copt - t$ we get 

    \begin{equation*}
        \H_f \leq t + (\copt - t) + (\Sigma_\alpha + \Sigma_\beta) < \copt + \frac{\copt}{\varphi} = \varphi \copt.
    \end{equation*}

    Thus for the remainder of this case we'll assume that $\tau_a$ exists. Now suppose that $\beta^{t_a} = \emptyset$. This means that from $t$ onward, $\H_f$ does not get any contribution from $\F$. Furthermore, since some task in $\S_t \cap \A$ is not eligible we know that $t < \frac{\copt}{\varphi}$. Thus we have 

    \begin{equation}
        \H_f \leq t + \underbrace{\copt}_{\text{All work comes from $\A$}} \leq \frac{\copt}{\varphi} + \copt = \varphi\copt.
    \end{equation}

    Thus for the remainder of this case we'll assume that $\beta^{t_a} \neq \emptyset$.
    Now consider some ineligible task $\tau_i \in \S_{t^-} \cap \A$. We have that $t_i \leq t_a$ and that $\beta^{t_a} \subseteq \beta^{t_i} \neq \emptyset$. We can now invoke \cref{lem:early_eligible} to bound $t$. Concretely we have $t \leq \max(t_i, \Sigma_\alpha, \varepsilon + \Sigma_\beta - \Sigma_{\beta^{t_i}})$. By assumption the maximum cannot resolve to $t_i$ since $t_i < t$. 

    If the maximum resolves to $\Sigma_\alpha$ we have 

    \begin{equation}
        \H_f 
        \leq t + \underbrace{(\copt + \Sigma_{\beta^{t}})}_{\text{maximum work from $\A$ and $\F$}} \leq \Sigma_\alpha + (\copt + \Sigma_\beta) < \copt + \frac{\copt}{\varphi} = \varphi\copt,
    \end{equation}
    where the last inequality comes from \cref{lem:boundalphabeta}.

    Otherwise the maximum resolves to $\varepsilon + \Sigma_\beta - \Sigma_{\beta^{t_i}}$ and we have 

    \begin{equation}
        \H_f \leq 
        t + (\copt + \Sigma_{\beta^{t}})
        \leq (\varepsilon + \Sigma_\beta - \Sigma_{\beta^{t_i}}) + (\copt + \Sigma_{\beta^{t}}) \leq (\varepsilon + \Sigma_\beta) + \copt \leq \frac{\copt}{\varphi} + \copt = \varphi\copt,
    \end{equation}
    where in the penultimate inequality we used the fact that $\Sigma_{\beta^{t_i}} \geq \Sigma_{\beta^{t}}$ because $t \geq t_i$.
\end{proof}

We are now ready to prove the main theorem, restated here, which requires some case analysis.

\begin{theorem*}[Restatement of \cref{thm:competitive}]
    $\H$ is a $\varphi$-competitive in the Eventually-committing task model.
\end{theorem*}

\begin{proof}

    As previously mentioned, it suffices to just consider the performance of $\H$'s fast machine. Furthermore by \cref{lem:lategaps} we can assume for the remainder of the proof that $\H$'s fast machine is not free for any $t \geq \rcopt$ (and is thus operating continuously for $t \geq \rcopt$). Importantly this implies that some task is running at $t = \rcopt$, we call this task the ``stuck task" $\tau_u$. We clarify that $\tau_u$ can be a task that starts at $t = \rcopt$, but not a task that finishes at $t = \rcopt$. There are two cases to consider.

    \paragraph{Case 1 ($s_u + t_u > \copt/\varphi$):}

    Define $\Aearly = \{\tau_i \in \A : t_i < \rcopt\}$ and $\Alate = \A \setminus \Aearly$. In this case, by largeness, $\tau_u$ must be a member of one of $\Aearly$, $\Alate$, or $\Fbig$. Our strategy will be to characterize the contribution of these sets to $\H_f$.
    
    Let $t < \rcopt$ be a time when all tasks in $\Fbig \cup \Aearly$ have already arrived; this must exist due to \cref{lem:nolatearrive}. Since $\C^t < \copt/\varphi$ all of these tasks must be on $\OPT(t)$'s fast machine at this time or else $\OPT(t)$'s runtime would be greater than $\frac{\copt}{\varphi}$. Therefore we have $\Sigma_{ \Fbig \cup \Aearly} \leq \copt/\varphi$.

    Notice that $\Alate$ can contribute at most $\copt - \rcopt$ to $\H_f$ since $\OPT$ must complete the work on its fast machine. Furthermore, by \cref{lem:nolatesmall}, we know that no tasks from $\Fsmall$ can contribute to $\H_f$ after $\rcopt$. Thus we have

    \begin{align}
        \H_f &\leq \rcopt + \Sigma_{\Fbig \cup \Aearly} + \Sigma_{\Alate} \\
        &\leq \rcopt + \frac{\copt}{\varphi} + \left(\copt - \rcopt\right) \\
        &\leq \left(1 + \frac{1}{\varphi}\right)\copt \\
        &\leq \varphi \copt.
    \end{align}

    \paragraph{Case 2 ($s_u + t_u \leq \copt/\varphi$):}

    In this case we have $\tau_u \in \Fsmall$ by \cref{lem:setprops}. Let $t^*$ be the time in which $\H$ starts $\tau_u$ on the fast machine. Since $\tau_u \in \Fsmall$, by \cref{lem:nolatesmall}, we have the strict inequality $t^* < \rcopt$ Now we can define some additional task sets (see \cref{fig:case-1}):

    \begin{itemize}
        \item $\Fbig'$: The set of tasks $\tau_i \in \Fbig$ such that $t_i \geq t^*$. Note that by \cref{lem:nolatearrive} we also know $t_i < \rcopt$.
        \item $\A_1$: The set of tasks $\tau_i \in \A$ such that $R(\copt/\varphi) > t_i \geq t^*$.
        \item $\A_2$: The set of tasks $\tau_i \in \A$ such that $t^* + f_u> t_i \geq R(\copt/\varphi)$.
        \item $A_3$: The set of tasks $\tau_i \in \A$ such that $t_i > t^* + f_u$.
    \end{itemize}

           \begin{figure}
    \centering
    \begin{tikzpicture}
    \path[use as bounding box,help lines,draw=none] (-2,-1) grid (6,1.5);

    \draw[thick, -stealth] (0, 0) -- (6, 0);

    \draw[thick,dashed] (0,0) -- (0,0.9);
    \draw[thick,dashed] (1,0) -- (1,0.9);
    \draw[thick,dashed] (3,0) -- (3,0.9);
    \draw[thick,dashed] (4,0) -- (4,0.9);
    \draw[thick,dashed] (5.6,0) -- (5.6,0.9);

    \node at (0,1.35) {$0$};
    \node at (1,1.35) {$t^*$};
    \node at (3,1.35) {$R(\frac{\copt}{\varphi})$};
    \node at (5.6,1.35) {$\copt$};

    \node at (2,0.4) {$\Fbig', \A_1$};
    \node at (3.5,0.4) {$\A_2$};
    \node at (4.8,0.4) {$\A_3$};

    \draw [decorate,decoration={brace,amplitude=5pt},rotate=90,thick] (-0.3,-4) -- (-0.3,-1);

    \node at (-2,0) {\text{$\H$'s fast machine}};

    \draw[draw=black,thick,fill=white, fill opacity=0.5] (1,-0.1) rectangle (4,0.1);
    \node at (2.5,-0.75) {$f_u$};

    \end{tikzpicture}  
        \caption{Main theorem, Case 2, arrival times of tasks relative to the stuck task $\tau_u$.}
        \label{fig:case-1}
    \end{figure}
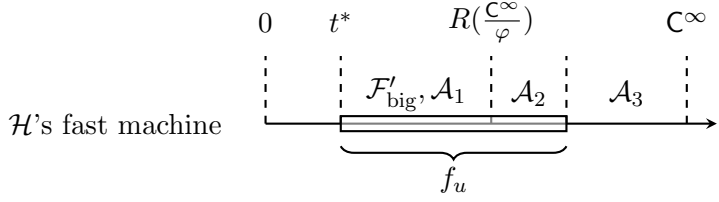

    Our strategy will be to characterize the contribution of all of the tasks assigned to $\H_f$'s fast machine.

     Since $\H$ prioritizes tasks with large $s_j + t_j$ there are no tasks from $\Fbig \cup \A$ on standby at $t^*$. Furthermore by \cref{lem:nolatesmall} we can assume there aren't any task from $\Fsmall$ in $S_{t^*}$. Thus we can conclude that there is no contribution from $\S_{t^*}$ to $\H_f$.

    Next, we recall that $t^* \geq \frac{f_u}{\varphi}$ because of the eligibility condition.

    Now we show that all tasks in $\Fbig'$ are assigned to the fast machine by $\OPT(t)$ for all $t^* \leq t \leq f_u + t^*$. To see this we note that if any task $\tau_i \in \Fbig'$ were assigned to one of $\OPT(t)$'s slow machine at time $t$ then $\H$ would assign $\tau_i$ to a slow machine at time $t$ because 

    \begin{equation}
    \frac{\overbrace{s_i + t}^{\tau_i \text{ on $\H$'s slow machine}}}{\C^{t}} \leq \frac{s_i + t^* + f_u}{s_i + t_i} \leq \frac{s_i + t_i + f_u}{s_i + t_i} \leq \frac{\copt/\varphi + \copt/\varphi^2}{\copt/\varphi} = \varphi,
    \label{eq:fbigfast}
    \end{equation}
    where in the second inequality we have $t^* \leq t_i$ by the definition of $\Fbig'$ and in the last inequality we use the facts that $f_u \leq \frac{\copt}{\varphi^2}$ and $s_i + t_i > \frac{\copt}{\varphi}$ from \cref{lem:setprops}.

    Note that all tasks in $\A$ are assigned to $\OPT$'s fast machine at all times after arrival. This is because they are on $\OPT$'s fast machine at the end of the TAP, and \cref{lem:noslowswap} tells us $\OPT$ never swaps tasks from the slow machine to the fast machine.

    Let $t_y < \rcopt$ be the arrival time of the last task in $\Fbig'\cup\A_1$. We therefore have,  
    
    \begin{equation}
        t^* + \underbrace{\Sigma_{\Fbig'} + \Sigma_{\A_1}}_{\text{on $\OPT(t_y)$'s fast machine}} \leq \C^{t_y} < \frac{\copt}{\varphi}.
        \label{eq:bounda1}
    \end{equation}

    From here we define

    \begin{equation}
        \Delta = \frac{\copt}{\varphi} - (t^* + \Sigma_{\Fbig'} + \Sigma_{\A_1}).
    \end{equation}

    Now we have to bound $\Sigma_{\A_2}$. For this we use the following claim.

    \begin{claim}
        If $\Fbig' = \emptyset$, then $H_f \leq \varphi\copt$.
    \end{claim}

    \begin{proof}
        Suppose that $\Fbig' = \emptyset$. Recall from \cref{lem:nolatearrive}, that no tasks from $\F$ will arrive after $\rcopt$. Furthermore, as we previously argued, there is no contribution to $\H_f$ from $\S_{t^*}$. Thus we actually have that there is no further contribution to $\H_f$ from $\F$.
        
        From $t^*$ onward no more than a total workload of $\copt - t^*$ can arrive from tasks from $\A$ since $\OPT$'s fast machine needs to do this work.        
        
        Thus since $\H$'s fast machine runs continuously after $t \geq t^* + f_u >\rcopt$, we have

        \begin{equation}
            \H_f = (t^* + f_u) + (\copt - t^*) \leq \frac{\copt}{\varphi^2} + \copt \leq \varphi \copt.
        \end{equation}

    \end{proof}

    Now we may assume that $\Fbig' \neq \emptyset$.
    
    Let $t_z \leq t^* + f_u$ be a time by which all tasks in $\A_2$ have arrived. As noted in \eqref{eq:fbigfast}, all of the tasks in $\Fbig'$ remain on $\OPT(t)$'s fast machine for $t^* \leq t \leq f_u + t^*$ and are therefore assigned to the fast machine by $\OPT(t_z)$. Thus at $t_z$, all of $\A_1, \Fbig'$, and $\A_2$ are assigned to the fast machine by $\OPT(t_z)$. Since $\H$ assigns none of the tasks $\tau_i \in \Fbig'$ to a slow machine, we have
    
    \begin{align}
        \varphi \leq \frac{s_i + t_z}{\C^{t_z}} \leq \frac{s_i + (t^* + f_u)}{\tC^{t_z}} \leq \frac{s_i + (t^* + f_u)}{t^* + \Sigma_{\Fbig'} + \Sigma_{\A_1} + \Sigma_{A_2}}.
    \end{align}

    Rearranging the first and last expression and using that fact that $t^* +f_{\Fbig'} + f_{\A_1} = \frac{\copt}{\varphi} - \Delta$ gives

    \begin{equation}
       \Sigma_{\A_2} \leq \frac{(t_i + s_i) - \copt}{\varphi} + \frac{(t^* - t_i)}{\varphi} + \frac{f_u}{\varphi}  + \Delta \leq \frac{f_u}{\varphi} + \Delta,
    \end{equation}
    where in the second inequality we used the fact that $t_i \geq t_*$ and that $\copt \geq s_i + t_i$ from \cref{lem:setprops}.

    Lastly, we have 
    
    \begin{equation}
        \Sigma_{A_3} \leq \copt - (f_u + t^*)
    \end{equation}
    because $\OPT$ has to finish all of these tasks on the fast machine. 
    
    Taken together, we can bound $\H_f$ as follows:

    \begin{align}
        \H_f&=t^* + \Sigma_{\Fbig'} + \Sigma_{\A_1} + \Sigma_{A_2} + f_u + \Sigma_{A_3}  \\
        &\leq \frac{\copt}{\varphi} - \Delta + \frac{f_u}{\varphi} + \Delta + f_u + \copt - (f_u + t^*) \\
        &\leq \left(1+\frac{1}{\varphi}\right)\copt \\
        &=\varphi \copt.
    \end{align}

\end{proof}

\section{Never-committing scheduler}
\label{sec:nevercommit}

Now we move onto the next considered model in which task assignment can be altered.

In this section we show that for any Never-committing scheduler $\M$ and $\varepsilon >0$, there exists a TAP $\T$ for which $\M$ cannot achieve a competitive of $1.5-\varepsilon$ on $\T$. We note that this matches the upper bound of~\cite{itcs_NSAW_2025}.

As a reminder, in this setting a scheduler $\M$ has the ability to restart a task on a different machine, so our proof will have to consider this possibility.

\begin{theorem}
    No online algorithm $\M$ can achieve a competitive ratio of $1.5 - \varepsilon$ for any $\varepsilon > 0$, in the Never-committing task model.
\end{theorem}

\begin{proof}
    Assume for the sake of contradiction that there exists an algorithm $\M$ achieving a competitive ratio of $1.5 - \varepsilon$ for all TAPs $\T$. We give a strategy for an adversary by generating a TAP for $\M$.

    First, let $k$ be any even integer such that $k \geq \frac{1}{\varepsilon}$.

    Consider a TAP $\T$ consisting of 
    \begin{itemize}
        \item $\tau_0 = (1, 1.5, 0)$,
        \item $\tau_i = (\frac{1}{k}, \infty, \frac{1.5i}{k}), \forall i \in [1,\frac k2 -1]$,
        \item $\tau_L = (0.75, \infty, 0.75)$,
    \end{itemize}
    where the use of $s_i, s_L = \infty$ just forces $\M$ to put $\tau_i$ on the fast machine; $\infty$ can be substituted by a large enough constant.
    
    To clarify, this is a TAP that an adversary has in mind, but the adversary can choose to truncate the TAP (not send any remaining tasks from $\T$) whenever they want to. 

    From this TAP description, we can easily compute the optimal offline completion time at each arrival time. Clearly $\C^0 = 1$ by assigning $\tau_0$ to the fast machine. For any $i \in [1, \frac{k}{2}-1]$ we have that 

    \begin{equation}
     \C^{t_i} = \max(\underbrace{1 +\frac{i}{k}}_{\text{work on fast machine}}, \underbrace{\frac{1.5i}{k} + \frac{1}{k}}_{\text{start $\tau_i$ at $t_i$}}) = 1 + \frac{i}{k} \text{ for $i < \frac{k}{2}$},   
    \end{equation}
    which is achieved by assigning all task that have arrived up to $t_i$ to the fast machine. This is apparent because if any of the tasks are assigned to a slow machine, then the completion time would be at least 1.5 which is greater than $1 + \frac{i}{k}$ for $i < \frac{k}{2}$.

    Finally, upon the arrival of $\tau_L$, we have 
    \begin{equation}
     \C^{t_l} = \C^{0.75} = 1.5,   
    \end{equation}
    which is achieved by assigning $\tau_0$ to a slow machine and all other tasks to the fast machine. Now we use these completion times to complete the argument. Our basic strategy will be to show that $\M$ can never put $\tau_0$ on a slow machine.

    Notice that upon the arrival of $\tau_0$, $\M$ cannot put this task on the slow machine, or else the adversary will simply not send any additional tasks from $\T$ and the competitive ratio will be 1.5.

    Upon the arrival of any task $\tau_i$, $\M$ will have to assign $\tau_i$ to the fast machine because its slow runtime is infinite. Furthermore, if $\M$ moves $\tau_0$ to a slow machine at this time the competitive ratio would be given by at least

    \begin{equation}
        \frac{s_0 + t_i}{\C^{t_i}}  = \frac{1.5 + 1.5\frac{i}{k}}{1 + \frac{i}{k}} = 1.5.
    \end{equation}    

 Thus if $\M$ moves $\tau_0$ to a slow machine, then the adversary can simply stop sending tasks and $\M$ will have failed to meet a competitive ratio of $1.5 - \varepsilon$.

    Finally, consider the arrival of $\tau_L$. If $\M$ moves $\tau_0$ to the slow machine at this time, the competitive ratio will be at least 

    \begin{equation}
        \frac{t_l + s_0}{\C^{t_l}} = \frac{0.75 + 1.5}{\C^{0.75}} = \frac{2.25}{1.5} = 1.5,
    \end{equation}
    which does not achieve the desired competitive ratio of $1.5 - \varepsilon$.

    This means that $\tau_0$ must be on $\M$'s fast machine. Moreover, all other tasks must also be on $\M$'s fast machine because their slow runtimes are infinite. Thus all the tasks must be on the fast machine for $\M$, yielding a total completion time of
    
    \begin{equation}
        1 + \left(\frac k2 -1\right)\left(\frac{1}{k}\right) + 0.75 = 2.25 - \left(\frac{1}{k}\right),
    \end{equation}
    and thereby giving a competitive ratio of 

    \begin{equation}
        \frac{2.25 - \frac{1}{k}}{1.5} \geq 1.5 - \frac{\varepsilon}{1.5} > 1.5 - \varepsilon.
    \end{equation}

    This is a contradiction since $\M$ is a $(1.5 - \varepsilon)$-competitive online algorithm. Thus there does not exist any $(1.5-\varepsilon)$-competitive algorithms.
\end{proof}

\bibliographystyle{plain}
\bibliography{references}

\appendix

\end{document}